\documentclass[lettersize,journal]{IEEEtran}
\usepackage{amsmath,amsfonts}
\usepackage{algorithmic}
\usepackage{algorithm}
\usepackage{array}
\usepackage[caption=false,font=normalsize,labelfont=sf,textfont=sf]{subfig}
\usepackage{textcomp}
\usepackage{stfloats}
\usepackage{url}
\usepackage{verbatim}
\usepackage{graphicx}
\usepackage{cite}
\usepackage{tikz}
\usepackage{hyperref}
\usetikzlibrary{shapes.geometric}

\usepackage{amssymb, amsthm}
\usepackage{multirow}
\usepackage{mathtools} 
\usepackage{xfrac} 

\newtheorem{lemma}{Lemma}
\newtheorem{definition}{Definition}

\DeclareMathOperator*{\PROBLEM}{VMCAP}
\DeclareMathOperator*{\maximize}{maximize}

\newcommand{\lingraph}{%
    \begin{tikzpicture}[scale=0.5, baseline=0, thick]
    \node[shape=rectangle,draw=black,fill,inner sep=1pt] (A) at (0pt,5pt){};
    \node[shape=rectangle,draw=black,fill,inner sep=1pt] (B) at (12pt,5pt){};
    \path (A) edge (B);
    \end{tikzpicture}%
}

\newcommand{\sqgraph}{%
    \begin{tikzpicture}[scale=0.5, baseline=0, thick]
    \node[shape=rectangle,draw=black,fill,inner sep=1pt] (A) at (0pt,0pt){};
    \node[shape=rectangle,draw=black,fill,inner sep=1pt] (B) at (12pt,0pt){};
    \node[shape=rectangle,draw=black,fill,inner sep=1pt] (C) at (12pt,12pt){};
    \node[shape=rectangle,draw=black,fill,inner sep=1pt] (D) at (0pt,12pt){};
    \path (A) edge (B);
    \path (B) edge (C);
    \path (C) edge (D);
    \path (D) edge (A);
    \end{tikzpicture}%
}

\newcommand{\sqstar}{%
    \begin{tikzpicture}[scale=0.5, baseline=0, thick]
    \node[shape=rectangle,draw=black,inner sep=1.2pt] (O) at (15pt,12pt){};
    \node[shape=rectangle,draw=black,fill,inner sep=1pt] (A) at (0pt,0pt){};
    \node[shape=rectangle,draw=black,fill,inner sep=1pt] (B) at (10pt,0pt){};
    \node[shape=rectangle,draw=black,fill,inner sep=1pt] (C) at (20pt,0pt){};
    \node[shape=rectangle,draw=black,fill,inner sep=1pt] (D) at (30pt,0pt){};
    \path (O) edge (A);
    \path (O) edge (B);
    \path (O) edge (C);
    \path (O) edge (D);
    \end{tikzpicture}%
}

\begin{document}

\title{Efficient calculation of available space for multi-NUMA virtual machines}

\author{
  Andrei~Gudkov\textsuperscript{1},
  Elizaveta~Ponomareva\textsuperscript{2},
  Alexis~Pospelov\textsuperscript{3} \\
  \par\vspace{5pt}\textit{Huawei Technologies Company Ltd, Lomonosov Research Institute, Moscow, 121099, Russian Federation}
}

\maketitle

\footnotetext[1]{Email: \href{mailto:gudokk@gmail.com}{gudokk@gmail.com}}
\footnotetext[2]{Email: \href{mailto:lizaveta@yandex.ru}{lizaveta@yandex.ru}}
\footnotetext[3]{Email: \href{mailto:alexis.pospelov@gmail.com}{alexis.pospelov@gmail.com}}

\begin{abstract}
Increasing demand for computational power has led cloud providers to employ
multi-NUMA servers and offer multi-NUMA virtual machines to their customers.
However, multi-NUMA VMs introduce additional complexity to scheduling algorithms.
Beyond merely selecting a host for a VM, the scheduler has to map virtual NUMA
topology onto the physical NUMA topology of the server to ensure optimal VM
performance and minimize interference with co-located VMs.
Under these constraints, maximizing the number of allocated multi-NUMA VMs on a host
becomes a combinatorial optimization problem.
In this paper, we derive closed-form expressions to compute the maximum number of VMs
for a given flavor that can be additionally allocated onto a physical server.
We consider nontrivial scenarios of mapping 2- and 4-NUMA symmetric VMs to 4- and
8-NUMA physical topologies.
Our results have broad applicability, ranging from real-time dashboards
(displaying available cluster capacity per VM flavor)
to optimization tools for large-scale cloud resource reorganization.

\vspace{0.5em}
\end{abstract}

\section{Introduction}\label{section:introduction}

Multicore CPUs are commonly found in all modern computers ranging from inexpensive
hand-held devices to sophisticated servers used in high-performance and cloud computing.
A large number of cores enables the parallel execution of multiple threads with access to shared memory.
Since the number of cores, volume, and bandwidth of memory of a single CPU are limited,
high-end systems make use of Non-Uniform Memory Access (NUMA) architecture to further
improve these characteristics.
Such systems are built with multiple (2-8) NUMA nodes, each consisting of a CPU and memory
modules attached to it (Fig. \ref{fig:intro1}).
NUMA nodes are connected together using high-speed peer-to-peer communication links,
allowing a thread executing in one NUMA node to transparently access memory located
in another NUMA node.
From the perspective of an ordinary application, such a system looks like a single large multicore
system with flat memory address space.
Accessing remote memory is transparent, but it incurs additional latency proportional to the number
of hops in the shortest path to the target NUMA node and stresses the bandwidth of inter-NUMA links.
The majority of applications benefit if all their threads and memory
are bound to a single NUMA node\cite{cheng2013evaluation}\cite{schildermans2021virtualization}
to increase data locality.
In some cases, when an application cannot fit on a single NUMA node
or when it requires exceptionally high memory throughput, it may be beneficial
to distribute the memory and threads of the application among multiple -- but close -- NUMA nodes.

\begin{figure}[!t]
\centering
\includegraphics[width=0.76\columnwidth]{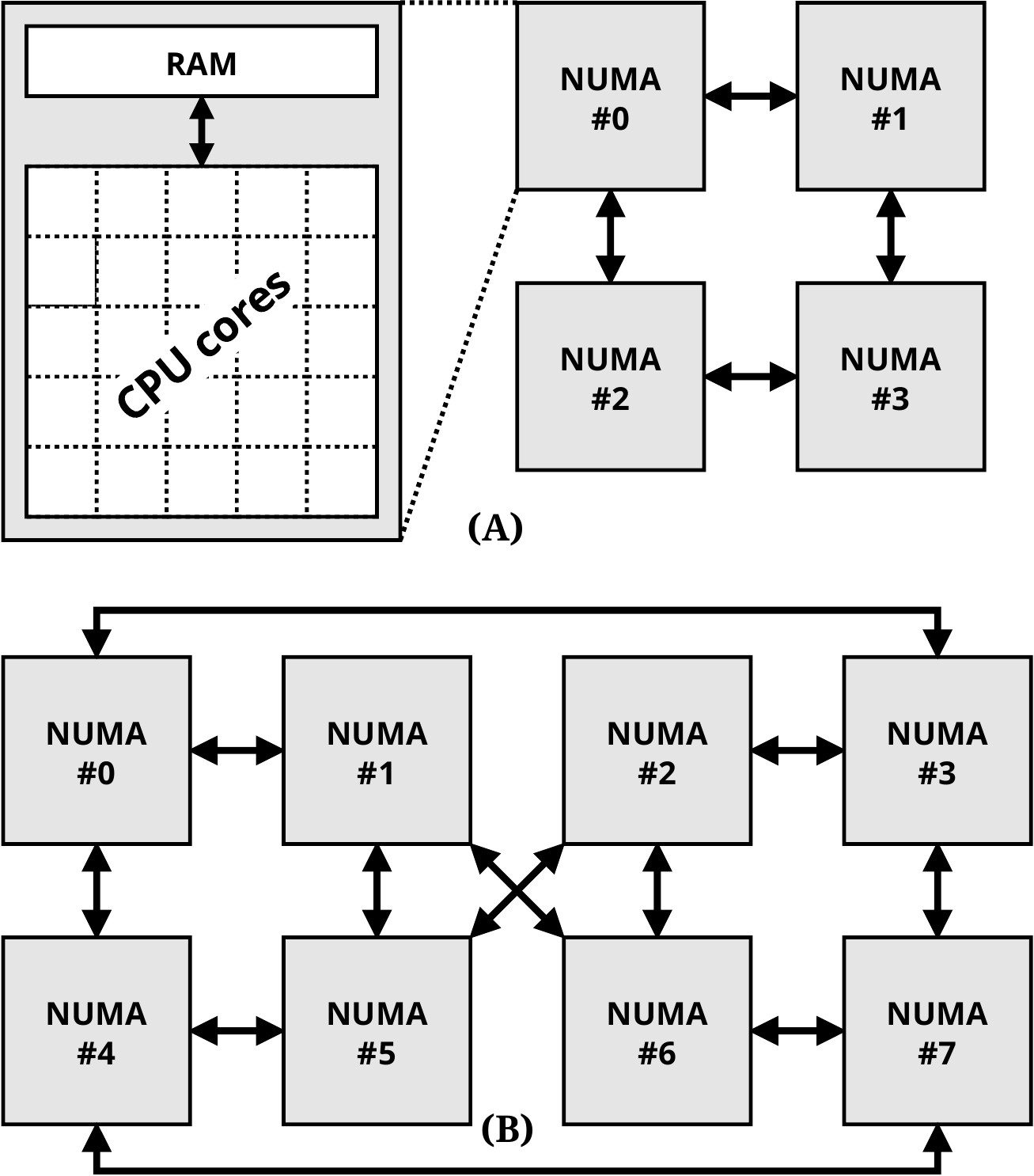}
\caption{Possible (A) 4-NUMA and (B) 8-NUMA topologies}
\label{fig:intro1}
\end{figure}

In the domain of cloud computing, physical resources of each server are divided into multiple
isolated virtual machines (VMs).
More specifically, customers order virtual machines with predefined characteristics,
such as the number of CPU cores (vCPU) and the amount of memory (RAM),
which are colloquially called \textit{flavors}.
When the cloud provider receives such a request, it searches for a server with sufficient free
resources and launches a VM there, reserving requested resources for the duration of the VM's lifespan.
These resources are released when the customer cancels the VM.
From the user's perspective, a virtual machine behaves like a fully functional computer that has
its own set of vCPU cores, memory, NUMA topology (vNUMA), and operating system.
At the same time, the operating system of the host views each VM simply as a process with one thread per vCPU core.

To cover a variety of user needs, cloud providers offer VMs of different sizes and vNUMA topologies.
The simplest vNUMA topology consists of a single vNUMA node, which must be
mapped to one of the pNUMA nodes that has sufficient free resources.
An example of a more complicated 2-vNUMA topology is a \lingraph\;graph, where both vNUMA nodes require
the same number of vCPU cores and the same amount of RAM.
It has to be mapped to a pair of distinct pNUMA nodes connected by a direct link and
each with sufficient free resources.
Violation of this rule, such as asymmetric resource allocation between two pNUMA nodes,
mapping both vNUMA nodes to the same pNUMA node, or placing them to indirectly connected pNUMA nodes
would create undesirable performance variation between VMs of the same flavor and is therefore avoided.
An even more complicated vNUMA topology is that of a square \sqgraph,
again with identical resource demand per vNUMA node.

Cloud providers are often interested in knowing how many additional VMs of a specific flavor 
can be started in an already partially occupied server (Fig. \ref{fig:intro2}).
This problem has numerous applications:
\begin{itemize}
  \item Summed over all servers in the cluster, it directly serves as the capacity
        metric that indicates how many VMs of a given flavor can additionally be sold \cite{verma2014evaluating}.
  \item It can be used as a subroutine of the scheduler that selects the server and pNUMA nodes
        where to start a new VM.
  \item In a more challenging scenario, it may be required to optimize the whole cluster state
        with respect to the capacity metric, i.e., to reassign already existing VMs (using live migration)
        inside and between servers in order to maximize the number of VMs of a given flavor
        that can additionally be sold \cite{zhu2024phecon}\cite{sun2024numa}.
        The challenge here is that, since this type of optimization is NP-hard and often involves
        the application of local search methods, the capacity metric has to be evaluated millions of times
        for different inputs during a single optimizer run.
\end{itemize}

For all combinations of pNUMA/vNUMA topologies except the simplest cases
there is no straightforward algorithm for solving this problem.
As we show in the next section, it can be approached with graph-theoretic algorithms
or formulated as an integer linear programming (ILP) model
and solved with one of many available ILP solvers.
However, both approaches are undesirable in practice because of complexity
or reliance on third-party tools.
Another alternative, which may appear attractive in latency-sensitive online components,
is to precompute values for all possible inputs into a lookup table (LUT).
But it is viable only for topologies with a tiny number of pNUMA nodes.
For example, a LUT for a system with four pNUMA nodes of 100 cores each would already require
on the order of $100^4$ entries.
In this article, guided by the fact that the number of real-world pNUMA/vNUMA topology
combinations is limited, we present and formally prove simple formulas for all of them.
These formulas accept the remaining capacities of each pNUMA as inputs
and require nothing beyond simple algebraic operators.
Searching for specific per-VM placements is outside the scope of this article,
although most of our formulas suggest straightforward algorithms to accomplish this.

\begin{figure}[!t]
\centering
\includegraphics[width=\columnwidth]{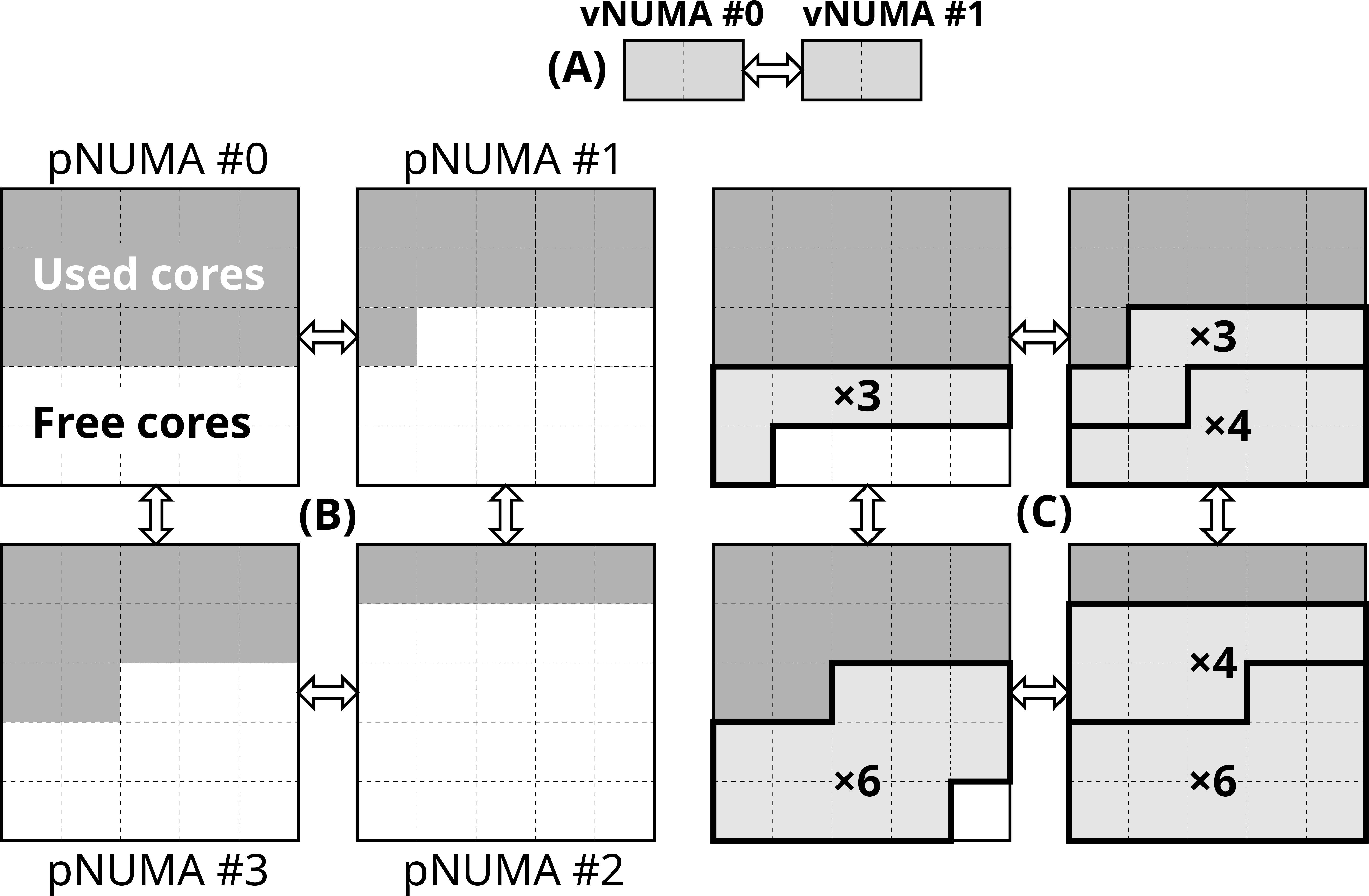}
\caption{
  Optimal solution for placing 2-vNUMA VMs (A) into partially filled 4-pNUMA server (B) is demonstrated in (C).
  3 VMs are placed in a pair of top pNUMA nodes, 4 VMs -- in a pair of right pNUMA nodes, and 6 VMs --
  in a pair pair of bottom pNUMA nodes.
  In total there is space for $3+4+6=13$ VMs.
  Some of the cores remain unused.
  Only CPU resource is considered.
  Geometric position of CPU cores within pNUMA nodes is irrelevant.
}
\label{fig:intro2}
\end{figure}

\section{Formal definition}\label{section:formal}

pNUMA and vNUMA topologies can be naturally represented with undirected graphs,
which we will denote as $G$ and $g$ respectively.
Graph vertices are the NUMA nodes, while edges represent links between them.
$g$ is assumed to be connected and has at least two vertices.
Following conventional notation, let $V(\cdot)$ denote the set of vertices of a graph,
$E(\cdot)$ -- the set of all edges, and $A(i)$ -- the set of edges adjacent to vertex $i$.
Furthermore, let the pNUMA graph $G$ be capacitated at vertices with capacity vector $b$:
$b_i \in \mathbb{Z}_{\geqslant 0} \; \forall i\!\in\!\mathrm{V}(G)$.
The capacity of vertex $i$ is the number of vNUMA nodes that the pNUMA node corresponding to $i$ can accommodate.
It is precomputed as the minimum across all resources for which the cloud provider guarantees dedicated allocation.
Most often these resources are the number of CPU cores and the amount of RAM, but may occasionally include
less tangible resources, such as L3 cache\cite{shahrad2021provisioning}:
\begin{align}
b_i = \min\limits_{\mathmakebox[5.4em][r]{r \in \left\{\mathrm{CPU}, \mathrm{RAM}, ...\right\}}}
{\left\lfloor\!\frac{\text{free capacity of $i$ in resource $r$}}{\text{demand of single vNUMA for $r$}}\!\right\rfloor}
\end{align}

Now we are ready to formulate the problem in a graph-theoretic sense:
\begin{definition}
Given pNUMA graph $G$, vNUMA graph $g$, and capacity vector $b$ imposed on the vertices of $G$,
$\PROBLEM(G, g, b)$ is the maximum number of times graph $g$ can be simultaneously mapped
to subgraphs of $G$ in such a way that each vertex $i$ of $G$ appears in at most $b_i$ mappings.
\end{definition}

In the most general case, when $g$ and $G$ are not known in advance, the problem
is at least as hard as the subgraph isomorphism problem, which is NP-complete.
However, in the case of $g=\lingraph$ the problem degenerates into the maximum
cardinality $b$-matching problem\cite{Korte_Vygen_2018}\cite{ferdous2021algorithms}:
the goal is to match pairs of adjacent vertices as many times as possible with
no restriction on the number of times each edge is used,
but with the restriction that each vertex $i$ is matched at most $b_i$ times.
This problem admits a polynomial-time solution\cite{gabow1983efficient},
although the algorithm is complicated.

Within the scope of cloud computing the problem needs to be solved only for specific,
small graphs (table \ref{table:combinations}).
This observation allows us to enumerate the set $P(G,g)$ of all possible subgraphs
of $G$ isomorphic to $g$, which results in the following straightforward ILP formulation:
\begin{align}
&\PROBLEM(G, g, b) = \maximize\limits_{\left\{x_p\right\}} \sum\limits_{p}{x_p} & \\
&\hspace{0.5em}\mathrm{s.t.}\hspace{1em} \sum\limits_{\mathmakebox[2em][l]{p:\,i \in V(p)}}{x_p} \leqslant b_i & \forall i \in V(G) \\
&\hspace{0.5em}\phantom{\mathrm{s.t.}}\hspace{1em} x_p \in \mathbb{Z}_{\geqslant 0} & \forall p \in P(G,g)
\end{align}

In this formulation $\left\{x_p\right\}$ are the variables that
decide how many times $g$ is mapped to each of the possible subgraphs.
One way to interpret this ILP formulation is as a multiple-choice vector knapsack problem.
The number of knapsack dimensions is $\left| V(G) \right|$; there is an infinite supply
of identical items with $\left|P\right|$ binary vectors representing the alternatives
of how a single item can be placed; and the goal is to place as many items into the knapsack
as possible without exceeding per-dimension capacities $b_i$.

Table \ref{table:combinations} lists non-trivial combinations of $G$ and $g$ in common use.
The set of possible $G$ graphs varies depending on CPU manufacturer.
4-pNUMA topologies are typically $K_4$ (Intel 4S-3UPI CPUs, Huawei Kunpeng 920\cite{xia2021kunpeng})
due to the minimum possible inter-node distance, but may alternatively be $C_4$ (Intel 4S-2UPI).
8-pNUMA topologies are based on "cubic" graphs: either enhanced cube $Q_{3,3}$\cite{tzeng1991enhanced}
(Intel 8S-4UPI Sapphire Rapids CPUs) or crossed cube $CQ_3$\cite{efe1992crossed}\cite{efe1994topological}
(Intel 8S-3UPI Skylake CPUs \cite{kumar2017skylake}).
It is notable that classical cube $Q_3$ is not commonly used, which can be explained by its poor
worst-case distance of 3 hops.
Both listed 8-NUMA topologies demonstrate worst-case distance of only 2 hops, which comes
at the cost of either increasing number of links per pNUMA node from 3 to 4 ($Q_{3,3}$)
or by allowing $G$ to be non-bipartite ($CQ_3$).

Connectivity in AMD systems, in contrast, is a hierarchy of up to three layers -- sockets, quadrants,
core/cache dies -- with cores located in the bottom-most layer.
AMD provides a way to configure what is considered a pNUMA node\cite{amd2021socket}.
For example, a cloud administrator may configure each quadrant to be a separate pNUMA node while prohibiting
scheduling vNUMA nodes of a single VM onto different sockets because of high inter-socket latency.
Such a setup can be modeled as two independent stars \sqstar. 
Since star graph is equivalent in terms of connectivity to a complete graph,
maximizing the number of scheduled $k$-vNUMA VMs to $n$-leaf star topology is equivalent
to solving two independent $\PROBLEM(K_n, K_k, b)$ problems, one per socket.
In section \ref{section:kn_kk} we demonstrate an approach to derive formulas for computing
$\PROBLEM(K_n, K_k, b)$ for arbitrary $n$ and $k$.
Already listed scenarios of scheduling $K_2$ and $K_3$ into $K_4$ are special cases of this general problem.

\begin{table*}[t]
\centering
\begin{tabular}{c}
\includegraphics[width=\textwidth]{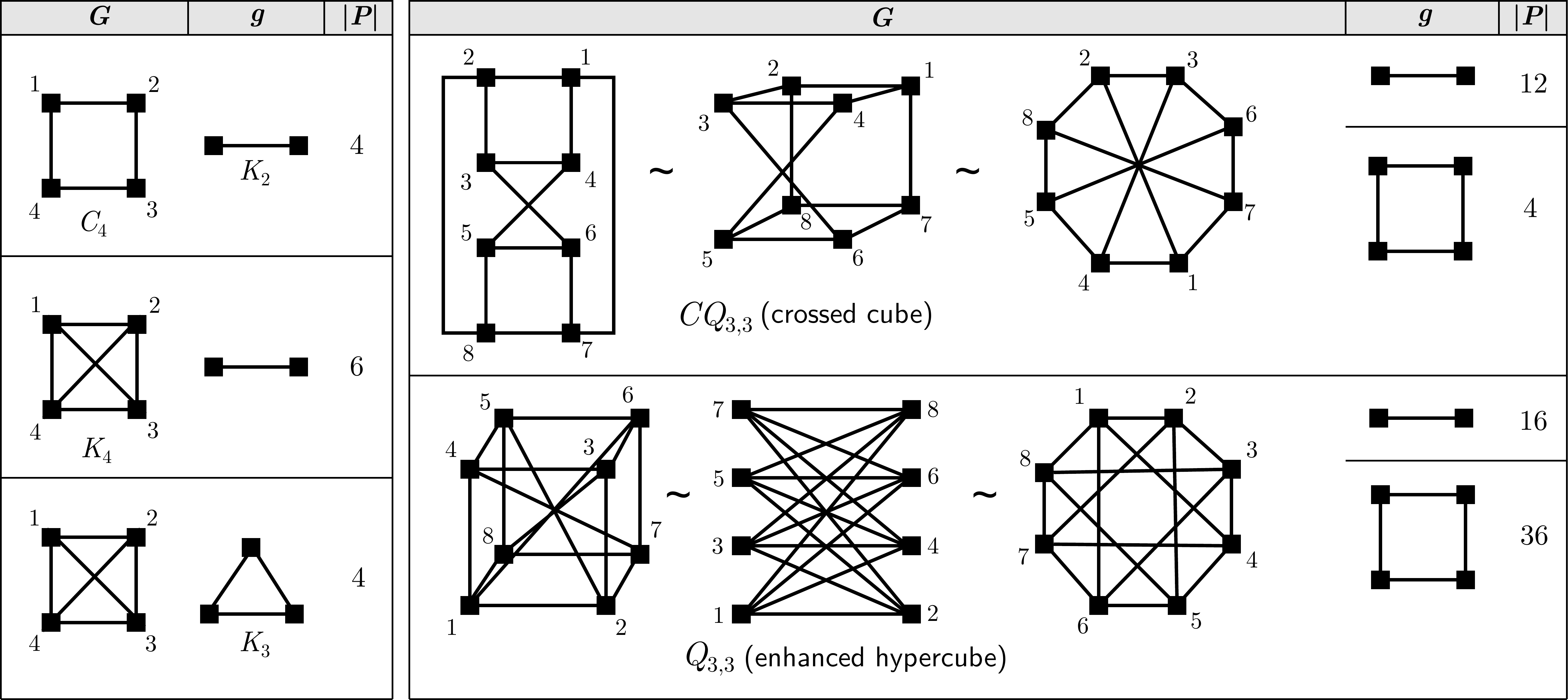}
\end{tabular}
\caption{The table lists seven considered combinations of $G$ (pNUMA topology) and $g$ (vNUMA topology).\\
         $G_1 \sim G_2$ indicates a graph isomorphism.
         $|P|$ column states the number of possible ways $g$ can be embedded into $G$.}
\label{table:combinations}
\end{table*}

The following sections consider each of the listed combinations of $G$ and $g$ separately.

\section{\texorpdfstring{$\PROBLEM(C_4, K_2, b)$}{VMCAP(C4, K2, b)}}\label{section:c4_k2}

Any $\PROBLEM(G,g,b)$ problem with $g\!=\!K_2$ can be converted into the classical maximum cardinality
matching problem (MCM) using the following reduction\cite{tutte1954short}.
Create an uncapacitated graph $G'$ such that each vertex $i$ of $G$ generates $b_i$
distinct vertices in $G'$, and two vertices in $G'$ are connected by an edge if and only if their
generating vertices in $G$ are connected:
\begin{align*}
& V(G') = \left\{i^{(k)} \;\;\forall i \in V(G),\; k=1..b_i\right\}\!, \\
& E(G') = \left\{(i^{(k_1)}, j^{(k_2)})\!: (i,j) \in E(G)\right\}\!.
\end{align*}

Thus $G'$ is a graph with $\sum_{i \in V(G)}{b_i}$ vertices
and $\sum_{(i,j) \in E(G)}{b_i\!\cdot\!b_j}$ edges.
It is clear that solving $\PROBLEM$ problem is equivalent to searching for the maximum
cardinality vertex matching in $G'$.
Since the size of $G'$ depends on the $b$ values, direct application of MCM algorithms
is inefficient because it leads to pseudo-polynomial running time in terms
of the size of the original graph $G$.
The inverse "compressing" transformation is equally valid:
\begin{lemma}\label{lemma:vertex-merging}
If there exists a subset of vertices $V' \subseteq V(G)$ such that
$A(i)=A(j)$ and $(i,j) \notin E(G)$ for all $i,j \in  V'$, then
merging all vertices in $V'$ into a single vertex $k$ with capacity $b_k=\sum\limits_{i \in V'}{b_i}$
and removing duplicate edges yields a graph equivalent to $G$ with respect to $\PROBLEM(G,K_2,b)$ metric.
\end{lemma}
\begin{proof}
Applying the previously described reduction procedure independently to the original
graph $G$ and to the graph with merged vertices yields exactly the same uncapacitated $G'$
and thus the same MCM instance (Fig. \ref{fig:k2c4}).
\end{proof}

\begin{figure}[ht]
\centering
\includegraphics[width=\columnwidth]{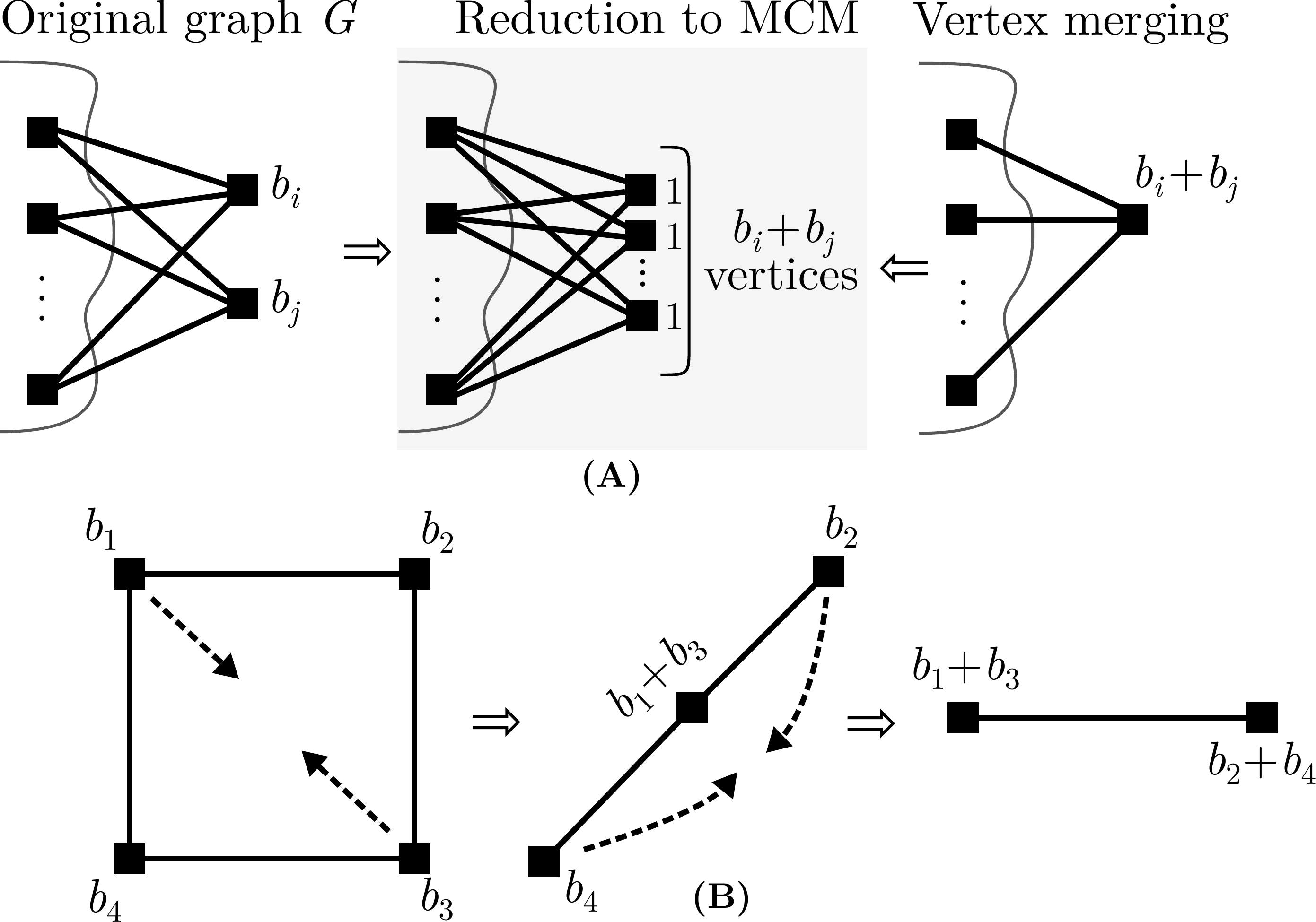}
\caption{(A) Illustration for Lemma \ref{lemma:vertex-merging};
         (B) the sequence of reductions used to derive $\PROBLEM(C_4, K_2, b)$.}
\label{fig:k2c4}
\end{figure}

To solve $\PROBLEM$ problem for $G=C_4$ and $g=K_2$ we apply Lemma \ref{lemma:vertex-merging} twice.
First, we merge vertices $1$ and $3$, and then we merge vertices $2$ and $4$,
yielding the following formula:
\begin{align} \label{formula:c4_k2}
\PROBLEM(C_4, K_2, b) = \min\left(b_1\!+\!b_3,\,b_2\!+\!b_4\right)\!.
\end{align}

\section{\texorpdfstring{$\PROBLEM(K_4, K_2, b)$ and $\PROBLEM(K_4, K_3, b)$}{VMCAP(K4, K2, b) and VMCAP(K4, K3, b)}}\label{section:kn_kk}
We begin by proving two formulas for the general problem of $\PROBLEM(K_n, K_k, b)$,
i.e. where both $G$ and $g$ are complete graphs of arbitrary sizes.
These formulas are based on the idea of comparing the trivial upper bound
$\left\lfloor\sfrac{1}{k}\sum{b_i}\right\rfloor$ with the largest vertex capacity $\max b_i$.
If this upper bound is greater than or equal to $\max b_i$, then it is attainable.
Otherwise, the problem can be reduced to $\PROBLEM(K_{n-1}, K_{k-1}, b)$
by dropping the vertex with the maximum capacity.

\begin{lemma}\label{lemma:knkkrec}
Given that $1 \leqslant k \leqslant n$:
\begin{align}\label{formula:knkkrec}
&\PROBLEM(K_n,K_k,b) = \notag\\
&\hspace{0.12in}\begin{cases}
  \left\lfloor\sfrac{1}{k}\Sigma b_i\right\rfloor\!,\hspace{0.6in}\text{if } \left\lfloor\sfrac{1}{k}\Sigma b_i\right\rfloor \geqslant \max b_i,\\
  \PROBLEM(K_{n-1},K_{k-1},b\setminus\max b_i),\;\text{otherwise.}
\end{cases}
\end{align}
\end{lemma}

\begin{proof}
Comparing $\left\lfloor\sfrac{1}{k}{\sum{b_i}}\right\rfloor$ with $\max b_i$
results in one of the three outcomes:

\textbf{Case 1}: $\left\lfloor\sfrac{1}{k}{\sum{b_i}}\right\rfloor = \max b_i$.
Consider the "densest" possible instance: the first $k$ vertices have capacity $\max b_i$,
the capacity of the next vertex is $0 \leqslant b_{k+1} < \max b_i$, and the capacities
of the remaining vertices are all zero.
Matching together the first $k$ vertices $\max b_i$ times produces the desired solution.
Now consider any other -- "sparser" -- instance.
Note that moving capacity from vertex $i$ to vertex $j$ of capacity $b_j \geqslant b_i$
cannot improve $\PROBLEM$ value -- it can only decrease it.
Relabel the vertices such that $b_1 \geqslant b_2 \geqslant b_3 \geqslant \dots$
and then shift all the capacity to the left while keeping per-vertex capacity limited
by the upper bound.
By doing so we have transformed the instance into the "densest" possible instance,
for which the upper bound is attainable.
Therefore, the upper bound is attainable for all instances.

\textbf{Case 2}: $\left\lfloor\sfrac{1}{k}\sum{b_i}\right\rfloor > \max b_i$.
By multiplying both sides by $k$ and using the inequality
$q \left\lfloor p/q \right\rfloor \leqslant p$ we get $\sum{b_i} > k \max b_i$.
Therefore, there are at least $k$ vertices with non-zero capacities.
Let us match the $k$ vertices with the largest capacities and decrement their capacities by $1$.
The left-hand side of the original inequality condition is reduced exactly by one,
while the right-hand side either decreases by one or remains unchanged.
Therefore the new instance is either case (1) or again case (2).
Repeating this process eventually terminates in (1) and, together with the produced matches,
results in overall $\PROBLEM$ of $\left\lfloor\sfrac{1}{k}{\sum{b_i}}\right\rfloor$.

\textbf{Case 3}: $\left\lfloor\sfrac{1}{k}\sum{b_i}\right\rfloor < \max b_i$.
Since the vertex with the maximum capacity has excess capacity compared to the upper bound,
without loss of generality we can use this vertex in all matches.
Therefore $\PROBLEM$ value is determined solely by the remaining vertices,
i.e., how many times $K_{k-1}$ can be embedded into the remaining $n-1$ vertices.
This is an instance of $\PROBLEM(K_{n-1}, K_{k-1}, b \setminus \max b_i)$.

Note that case (3) can occur only when $k > 1$.
If $k = 1$, then we have $\left\lfloor\sfrac{1}{1}\sum{b_i}\right\rfloor \geqslant \max b_i$,
which is either case (1) or case (2).
This proves that the sequence of computations according to the stated formula always
terminates by producing a value.
\end{proof}

We can further simplify the recurrent formula.
Let's denote by $b^{(i)}$ the $i$-th largest capacity: $b^{(1)} \geqslant b^{(2)} \geqslant \dots \geqslant b^{(n)}$.
Then according to formula \ref{formula:knkkrec} the value of $\PROBLEM(K_n, K_k, b)$ equals to the one of the values in the sequence
$a_r=\left\lfloor\frac{1}{k-r} \left(\sum_{i \leqslant n}{b_i} - \sum_{i \leqslant r}{b^{(i)}}\right)\right\rfloor$
for $0 \leqslant r < k$.
On the other hand, all $a_r$ values are upper bounds on $\PROBLEM$ value.
Indeed, if we relax $r$ largest capacities to infinity, then in the optimal
solution each match would use each of the corresponding $r$ vertices exactly once.
Therefore $\PROBLEM$ value will be determined by the number of times $K_{k-r}$
can be matched using the remaining $n-r$ vertices.
$a_r$ is the trivial upper bound on this value.
Combining these two facts -- that all $a_r$ are upper bounds on $\PROBLEM$ and that
at least one of them is attainable -- leads to the following simplified formula devoid of conditionals:

\begin{align}\label{formula:knkkmin}
&\PROBLEM(K_n,K_k,b) = \notag\\
&\hspace{0.5in}\min\limits_{0 \leqslant r < k} \left\lfloor\frac{1}{k-r}\left(\sum_{i \leqslant n}{b_i} - \sum_{i \leqslant r}{b^{(i)}}\right)\right\rfloor\!.
\end{align}

By expanding this formula for the target $n$ and $k$ we arrive at the following closed-form expressions:
\begin{align}
&\PROBLEM(K_4, K_2, b) = \min\begin{pmatrix*}[l]
  \left\lfloor\sfrac{1}{2}\sum{b_i}\right\rfloor,\\
  \sum{b_i} - b^{(1)} \\
  \end{pmatrix*} \hspace{0.1in}\text{and} \label{formula:k4k2}\\
&\PROBLEM(K_4, K_3, b) = \min\begin{pmatrix*}[l]
  \left\lfloor\sfrac{1}{3}\sum{b_i}\right\rfloor,\\
  \left\lfloor\sfrac{1}{2}\left(\sum{b_i} - b^{(1)}\right)\right\rfloor,\label{formula:k4k3}\\
  \sum{b_i}-b^{(1)}-b^{(2)} \\
  \end{pmatrix*}\!,\\
&\hspace{0.5in}\text{where $b^{(i)}$ is the $i$-th largest capacity.}\notag
\end{align}

As a final remark, let us provide some further theoretical insight into $\PROBLEM(G_n, G_k, b)$.
The sequence $\{a_r\}$ is in fact unimodal (see Appendix \ref{appendix:unimodal} for the proof).
Given a problem with large $k$, it would be natural to use
a relevant optimization -- such as golden-section search -- to find the minimum.
However, this would not improve the asymptotic runtime, as the latter is dominated
by the time it takes to find $k$ largest values inside $K_n$, which requires $\mathrm{\Omega}(n)$ time.
Using formula \ref{formula:knkkrec} appears to be a better choice.
Due to economical considerations it might be expected that capacities of the vertices
are well-balanced, i.e., only a few vertices have excess capacity.
Thus, the conditional statements inside formula \ref{formula:knkkrec} are beneficial:
they short-circuit the formula after computing only few $a_r$ terms.

\section{\texorpdfstring{$\PROBLEM(CQ_3, K_2, b)$}{VMCAP(CQ3, K2, b)}}\label{section:cq3_k2}

This scenario is the most complex among those considered in this study.
We attribute this complexity to the combination of $G$ being non-bipartite and its large size.


We will call the capacity of vertex $i$ \textit{normalized} if it is not
larger than the sum of capacities of its adjacent vertices:
$b_i \leqslant  \sum_{(i,j) \in E(G)}{b_j}$.
We will denote by $\overline{b_i}$ the capacities satisfying this property.
The following general lemma introduces a normalization procedure that
can be applied to a $\PROBLEM$ problem with any $G$ and $g$.
\begin{lemma}\label{lemma:normalization}
After a single round of replacing vertex capacities with
$\min(b_i, \sum_{(i,j) \in E(G)}{b_j})$ $\forall i \in V' \subseteq V(G)$
1) all vertex capacities in $V'$ become normalized and 2) $\PROBLEM$ value remains unchanged.
\end{lemma}

\begin{proof}
1) Suppose that the capacities of two adjacent vertices $i$ and $j$ were
both changed during normalization.
The capacity of $i$ could change only if
$b_i > b_j + \sum_{(i,k \ne j) \in E(G)}{b_k} \geqslant b_j$.
Similarly, it must be that $b_j > b_i$.
Both inequalities cannot hold simultaneously, leading to the conclusion that
if the capacity of some vertex $i$ was changed, then capacities of all its adjacent vertices
remained the same, which ensures that $i$ is now normalized.
Otherwise, if the capacity of some vertex $i$ remained the same after normalization,
then the capacity of each of its adjacent vertices could be reduced to no less than $b_i$,
proving again normalization property for $i$.
2) Recall from the problem statement that $g$ is a connected graph with at least two vertices,
i.e., there are no isolated vertices.
When embedding such $g$ into $G$, a vertex in $G$ can only be matched simultaneously
with at least one of its neighbours and thus cannot participate in more matchings
than the sum of capacities of its neighbours.
Therefore, the normalization procedure preserves the value of $\PROBLEM(G, g, b)$.
\end{proof}

\begin{lemma}\label{lemma:cq3_k2_l4}
Consider the ladder graph $L_4$ with vertex labeling as shown in Fig. \ref{figure:cq3_k2_l4},
\begingroup
\setlength\arraycolsep{0.5pt}
\begin{align}\label{formula:cq3_k2_l4}
\PROBLEM(L_4, K_2, b) = \min\begin{pmatrix*}[l]
\overline{b_1}&+&b_3&+&b_5&+&\overline{b_7}, \\
\overline{b_2}&+&b_4&+&b_6&+&\overline{b_8} \end{pmatrix*}\!.
\end{align}
\endgroup
\end{lemma}

\begin{figure}[ht]
\centering
\includegraphics[width=0.5\columnwidth]{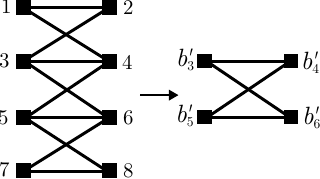}
\caption{Ladder graph $L_4$}
\label{figure:cq3_k2_l4}
\end{figure}

\begin{proof}
The value is clearly an upper bound on $\PROBLEM(L_4, K_2, b)$ since $L_4$ is bipartite and
normalization of vertex capacities doesn't lead to the loss of generality according
to Lemma \ref{lemma:normalization}.
Now let us construct a matching that achieves this upper bound.
After normalizing vertices $1$ and $2$ match them together as many times as possible
and, if one of them still has non-zero capacity left, match it with the remaining neighbour.
This procedure will reduce capacities of both vertices $1$ and $2$ to zero, since we have used normalized values.
Repeat the steps for vertices $7$ and $8$.
We are left with $C_4$ graph $(3,4,5,6)$ with some remaining capacities $b'_3, b'_4, b'_5, b'_6$.
Formula to compute its $\PROBLEM$ value was already derived in section
\ref{section:c4_k2} -- it is $\min(b'_3+b'_5, b'_4+b'_6)$, i.e.,
in a maximum matching either both $3$ and $5$ are zeroed, or both $4$ and $6$ are zeroed.
Altogether this means that we have constructed a feasible matching where either all
odd-indexed vertices or all even-indexed vertices are zeroed.
Cardinality of such matching cannot be lower than what is postulated by the lemma.
\end{proof}


\begin{figure}[ht]
\centering
\includegraphics[width=\columnwidth]{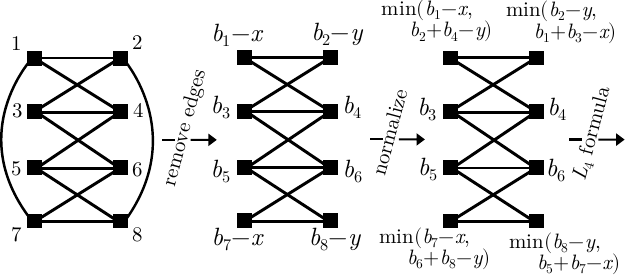}
\caption{Sequence of steps to derive optimization problem \ref{formula:cq3_k2_optx}}
\label{figure:cq3_k2}
\end{figure}

\begin{lemma}\label{lemma:cq3_k2}
$\PROBLEM(CQ_3, K_2, b)$ can be calculated as follows:
\begingroup
\setlength\arraycolsep{0.5pt}
\begin{align}
\label{formula:cq3_k2}
&\PROBLEM\left(CQ_3, K_2, b\right) = \min \begin{pmatrix*}[l]
b_1&+&b_3&+&b_5&+&b_7&-&\delta,\\
b_2&+&b_4&+&b_6&+&b_8&+&\delta,\\
b_2&+&b_3&+&b_4&+&b_5&+&b_7,\\
b_1&+&b_3&+&b_5&+&b_6&+&b_8,\\
b_2&+&b_4&+&b_5&+&b_6&+&b_7,\\
b_1&+&b_3&+&b_4&+&b_6&+&b_8
\end{pmatrix*}\!,\\
&\hspace{0.4in}\text{where }\delta = \lfloor\sfrac{1}{2}\cdot(\Sigma b_{2i+1} -\Sigma b_{2i})\rfloor \notag\\
&\hspace{0.4in}\text{clamped to } -\!\min(b_2,b_8) \leqslant \delta \leqslant \min(b_1,b_7). \notag
\end{align}
\endgroup
\end{lemma}

\begin{proof}
Let us redraw $CQ_3$ graph as depicted in Fig. \ref{figure:cq3_k2} on the left.
Then by introducing two variables $x,y \in \mathbb{Z}_{\geqslant 0}$ -- number of times edges $(1,7)$ and $(2,8)$ are matched respectively --
$\PROBLEM\left(CQ_3, K_2, b\right)$ can be expressed as the following optimization problem:
\begingroup
\setlength\arraycolsep{0.5pt}
\begin{align}
\label{formula:cq3_k2_optx}
&\PROBLEM\left(CQ_3, K_2, b\right) =\notag\\
&\hspace{0.1in}\maximize\limits_{x,y}\; x + y + \min \begin{pmatrix*}[c]
\min(b_1\!-\!x, b_2\!+\!b_4\!-\!y) + \\
b_3 + b_5 + \\
\min(b_7\!-\!x, b_6\!+\!b_8\!-\!y), \\[0.4em]
\min(b_2\!-\!y, b_1\!+\!b_3\!-\!x) + \\
b_4 + b_6 + \\
\min(b_8\!-\!y, b_5\!+\!b_7\!-\!x) \\
\end{pmatrix*}=\notag\\
&\hspace{0.1in}\maximize\limits_{x,y}\; \min \begin{pmatrix*}[l]
b_1&+&b_3&+&b_5&+&b_7&-&(x-y),\\
b_2&+&b_4&+&b_6&+&b_8&+&(x-y),\\
b_2&+&b_3&+&b_4&+&b_5&+&b_7,\\
b_1&+&b_3&+&b_5&+&b_6&+&b_8,\\
b_2&+&b_4&+&b_5&+&b_6&+&b_7,\\
b_1&+&b_3&+&b_4&+&b_6&+&b_8
\end{pmatrix*}\!.
\end{align}
\endgroup

To obtain this formula (see Fig. \ref{figure:cq3_k2}), we removed edges $(1, 7)$ and $(2,8)$ from $CQ_3$,
subtracted $x$ and $y$ from the capacities of their endpoints,
normalized the remaining capacities of the vertices $1$, $2$, $7$, $8$,
and plugged the latter values into the previously obtained formula \ref{formula:cq3_k2_l4}
for the remaining $L_4$ graph.
Next we expanded the expression by using the equality
$p+\min(q,r) =\allowbreak \min(p+q, p+r)$.
While doing so, we removed the terms $b_1+b_3+b_4+b_5+b_6+b_7-(x-y)$ and $b_2+b_3+b_4+b_5+b_6+b_8+(x-y)$
since they are not less than $b_1+b_3+b_5+b_7-(x-y)$ and $b_2+b_4+b_6+b_8+(x-y)$ respectively.
The optimal solution to the obtained six-term expression is achieved when $x-y$ is selected
(Fig. \ref{figure:cq3_k2_optx}) as integer value closest to
$\frac{1}{2}\left(\sum{b_{2i+1}}-\sum{b_{2i}}\right)$
and satisfying boundary conditions: $-\min(b_2,b_8) \leqslant x-y \leqslant \min(b_1, b_7)$.
\end{proof}

\begin{figure}[ht]
\centering
\includegraphics[width=\columnwidth]{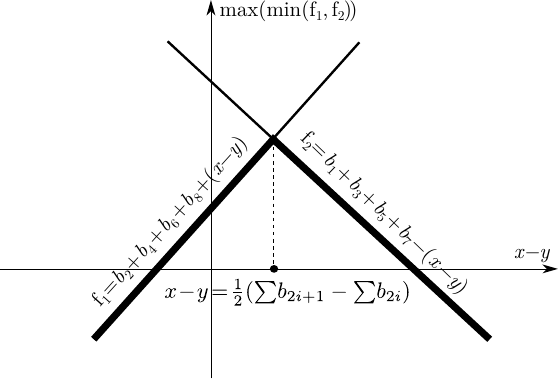}
\caption{Illustration of deriving $x\!-\!y$ value that optimizes expression \ref{formula:cq3_k2_optx}}
\label{figure:cq3_k2_optx}
\end{figure}

Note that the slopes of the two functions are symmetric.
Consequently, when computing $\delta$ in the formula \ref{formula:cq3_k2} it is irrelevant whether the floor
or ceiling function is used, or how their implementation behaves in the domain of negative numbers,
e.g. whether $\left\lfloor\frac{-5}{2}\right\rfloor=-2 \text{ or } -3$.

\section{\texorpdfstring{$\PROBLEM(CQ_3, C_4, b)$}{VMCAP(CQ3, C4, b)}}\label{section:cq3_c4}

There are four $C_4$ subgraphs in $CQ_3$ (Fig. \ref{fig:cq3_c4}).
They are connected into a cycle by means of the horizontal edges.
Consider some horizontal edge, e.g. $(3, 4)$.
It is shared by two $C_4$ subgraphs: $(2,1,4,3)$ and $(3,4,5,6)$.
Selecting one of these subgraphs as $g$ reduces the capacities of both $3$ and $4$ simultaneously by one,
and there is no way to reduce the capacity of $3$ or $4$ independently of each other.
The same reasoning applies to the other three horizontal edges.
This results in the equivalent transformation depicted in Fig. \ref{fig:cq3_c4} on the right.
Each vertex in the transformed graph represents one of the horizontal edges of $CQ_3$,
while each edge represents one of the $C_4$ subgraphs of $CQ_3$.
Selecting one of the $C_4$ in $CQ_3$ is equivalent to selecting the corresponding edge
in the transformed graph.
The latter problem is effectively $\PROBLEM(C_4, K_2, b)$, which has already been solved.
Thus, we have:

\begin{figure}[t]
\centering
\includegraphics[width=0.8\columnwidth]{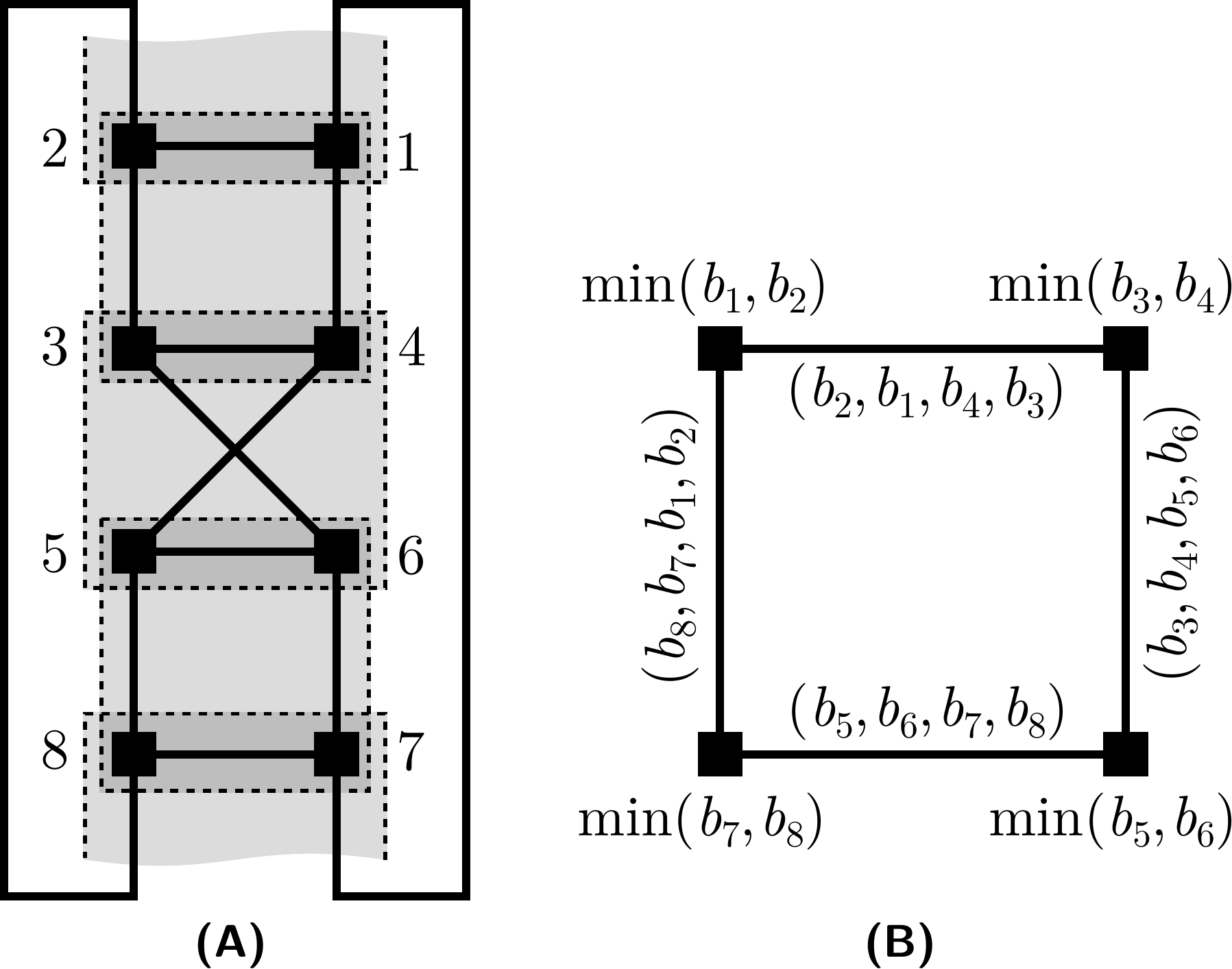}
\caption{(A) Cycles in CQ3 and (B) corresponding equivalent transformation}
\label{fig:cq3_c4}
\end{figure}

\begin{align}\label{formula:cq3_c4}
&\PROBLEM(CQ_3, C_4, b) =\notag\\
&\hspace{0.25in}\PROBLEM(C_4, K_2, [\min(b_1,b_2), \min(b_3,b_4),\notag\\
&\hspace{0.25in}\hspace{1.1in}\min(b_5,b_6), \min(b_7,b_8)]) =\notag\\
&\hspace{0.25in}\min\begin{pmatrix*}[l]\min(b_1,b_2)+\min(b_5, b_6),\\
\min(b_3, b_4)+\min(b_7,b_8)\end{pmatrix*}\!.
\end{align}

\section{\texorpdfstring{$\PROBLEM(Q_{3,3}, K_2, b)$}{VMCAP(Q33, K2, b)}}\label{section:q33_k2}

It can be observed that $Q_{3,3}$ is isomorphic to the complete bipartite graph $K_{4,4}$
(see table \ref{table:combinations}, middle graph).
One of its parts contains all odd-indexed vertices, while the other one -- all even-indexed vertices.
Thus, all odd-indexed vertices satisfy the requirements of Lemma \ref{lemma:vertex-merging}
and therefore can be merged into a single vertex of capacity $b_1+b_3+b_5+b_7$.
The same reasoning applies to even-indexed vertices, yielding the following formula:
\begingroup
\setlength\arraycolsep{0.5pt}
\begin{align} \label{formula:q33_k2}
\PROBLEM\left(Q_{3,3}, K_2, b\right) = \min
\begin{pmatrix*}[l]
b_1&+&b_3&+&b_5&+&b_7, \\
b_2&+&b_4&+&b_6&+&b_8
\end{pmatrix*}\!.
\end{align}
\endgroup

This formula is naturally generalized for arbitrary-sized complete bipartite graphs $K_{m,n}$.
If the vertices are labeled such that vertices $1 \mathinner{\ldotp\ldotp} m$ are one of its parts, then
\begin{align} \label{formula:knm_k2}
\PROBLEM(K_{m,n}, K_2, b) = \min\left(\sum\limits_{i=1}^{m}{b_i}, \sum\limits_{i=1}^{n}{b_{m+i}}\right)\!.
\end{align}
The earlier formula \ref{formula:c4_k2} for $\PROBLEM(C_4, K_2, b)$ is a special
case of formula \ref{formula:knm_k2} with $m\!=\!n\!=\!2$ since $C_4 \sim K_{2,2}$.

\section{\texorpdfstring{$\PROBLEM(Q_{3,3}, C_4, b)$}{VMCAP(Q33, C4, b)}}\label{section:q33_c4}

We again use the representation of $Q_{3,3}$ as bipartite $K_{4,4}$.
Any pair of distinct vertices in its left part together with any pair of distinct vertices
in its right part forms a $C_4$ subgraph, and there are no other $C_4$ subgraphs.
Therefore $\PROBLEM$ value is the minimum between the maximum number of vertex pairs
which can be selected from the left part and the maximum number of vertex pairs
which can be selected from the right part, subject to capacity constraints.
But the maximum number of vertex pairs in each part is exactly the $\PROBLEM(K_4, K_2, b)$ problem
since selecting a pair of distinct vertices in some part is equivalent to matching two vertices in $K_4$.
Hence we have:

\begin{align} \label{formula:q33_c4}
&\PROBLEM\left(Q_{3,3}, C_4, b\right) = \notag \\
&\hspace{0.4in}\min \begin{pmatrix*}[l]
\PROBLEM\left(K_4, K_2, \left[b_1, b_3, b_5, b_7\right]\,\right),\\
\PROBLEM\left(K_4, K_2, \left[b_2, b_4, b_6, b_8\right]\,\right)
\end{pmatrix*}\!.
\end{align}

\section{Related work}

The goal of computing $\PROBLEM$ in constant time using closed-form expressions
while supporting NUMA architectures beyond $K_n$ makes the problem at hand unusual,
positioned between applied research on NUMA-awareness and more theoretical works
on vertex matching problems.


In the context of NUMA-aware scheduling, many works are dedicated
to system-level engineering, such as investigating how NUMA scheduling policies
affect performance or reverse-engineering inter-NUMA cache coherence protocols.
However, few works focus on the discrete optimization aspects induced
by multi-NUMA topologies.
These few works broadly aim to solve one of the two following problems:
online scheduling and VM consolidation.
Online scheduling seeks to find a place for a VM (which may be multi-vNUMA)
coming from an ahead-unknown sequence of VMs with the aim to maximize
the number of successfully scheduled VMs.
Besides simple FirstFit- and BestFit-like heuristics, it can be approached
by using reinforcement learning \cite{sheng2022learning}
or complex value-driven heuristics \cite{ma2026topology}.
Consolidation is another type of problem.
It seeks to migrate VMs inside the cluster in order to optimize
target metric while keeping the number of migrations low.
For example, \cite{zhu2024phecon} provides a greedy heuristic to consolidate VMs while
optimizing the cluster-wide $\PROBLEM$ metric for specific flavor.
\cite{hu2020virtual} introduces a swarm optimization algorithm to reduce the number
of used servers, and its formulation admits even asymmetric multi-vNUMA VMs,
i.e. different resource demands for different vNUMA nodes.
\cite{sun2024numa} uses column generation to find an assignment of multi-NUMA VMs
that aims to maximize free space expressed as $\sum_f{w_f \cdot \PROBLEM_f}$,
where $w_f$ is the weight of the flavor $f$, and $\PROBLEM_f$ is its global per-cluster capacity.
However, all these works are restricted to $\PROBLEM(K_n, K_k, b)$ scenarios.
Balls-and-bins model introduced in \cite{rai2012generalized} is an interesting
unified approach to model many variants of topology-aware scheduling, including
NUMA nodes, although it does not provide any specific algorithms.


In general, cloud providers do not explicitly reveal NUMA-related information,
such as the number of vNUMA nodes of a given flavor, its vNUMA topology ($K_2$, $C_4$, etc)
or the policy of vNUMA-to-pNUMA mapping.
The first two items can be viewed after acquiring a VM by using built-in OS tools,
such as \textit{lstopo} in Linux.
An outdated (but still functioning) online resource \cite{instaguideInstancePrice}
provides historical information on Amazon flavors, including vNUMA topology.


Matching problems are the cornerstone of combinatorial optimization.
However, nearly all studies in this area are theoretical in nature.
It is also hard to find a good algorithm implementation for formulations
beyond classical variants of (uncapacitated) vertex matching problem,
a point also noted in \cite{muller2000implementing}.
This likely stems from the numerous variants of matching problems:
weighted vs. unweighted, vertex- vs. edge-capacitated, etc.
Nevertheless, some works can be related to the investigated problem,
in particular to $\PROBLEM(K_n, K_k, b)$.
In \cite{gabrovvsek2020multiple} the authors provide many heuristics 
to decompose a full $k$-partite edge-weighted graph into disjoint $k$-cliques.
This can be adopted for a problem of scheduling VMs in $K_n$ pNUMA topologies if
inter-pNUMA links are different in performance.
A well-known Havel-Hakimi algorithm \cite{hakimi1962realizability} constructs,
when possible, a graph with given vertex degrees.
It is based on the idea of incrementally matching vertex with the largest
unassigned degree with the remaining vertices, an idea that we have also
used in the proof of Lemma \ref{lemma:knkkrec}.

\section{Conclusion}
In this paper we formulated the $\PROBLEM$ problem to compute the maximum number
of multi-NUMA virtual machines which can be scheduled into a multi-NUMA server.
We enumerated practical combinations of pNUMA and vNUMA topologies
and formally proved closed-form formulas for all of them.

We hope that these formulas will find application beyond the scope of cloud computing.
We suppose that they may be used within OS kernels for multithread scheduling,
since kernel-space is considered a "sterile" environment with restrictions on the
complexity of algorithms and is even devoid of floating-point arithmetic.
Solving $\PROBLEM(K_n, K_k, b)$ for large $n$ is also a recurring problem in distributed computing,
where the goal is to spread replicas of some service (such as a distributed database
or LLM inference) across different racks to ensure fault tolerance.

Assuming that NUMA topologies will evolve, it would be interesting to investigate
the possibility of deriving $\PROBLEM$ formulas automatically for a given combination of $G$ and $g$.
We conjecture that at least for $g=K_2$ it is possible to do so, probably by checking
progressively longer odd-length paths for augmentation.
However, since automatically derived formulas may be enormously large,
manual derivation of $\PROBLEM$ formulas will likely remain preferable.

\section{Appendix}

\subsection{Unimodality of partial means of monotonous sequences}
\label{appendix:unimodal}
\begin{lemma}
Let $b_1\leqslant b_2 \leqslant \ldots \leqslant b_n$ be a non-decreasing sequence of numbers.
Denote by $\left\{S_m\right\}$ its partial sums:
$$
S_m = \sum^m_{i=1}b_i,\ \ \forall\, m \in [1\,..\,n]\!.
$$
Then for any integer $k \in [1\,..\,n]$ the sequence $\left\{a_i\right\}$ of partial means
$$
a_i = \frac{S_{n - i}}{k-i},\ \ \forall\, i \in [0\,..\,k\!-\!1]
$$
is unimodal with unique global minimum.
\end{lemma}

\begin{proof}
Suppose that for some $i$ the inequality $a_{i-1} \leqslant a_i$ holds:
$$
(k - i) \cdot S_{n - i + 1} \leqslant (k - i + 1) \cdot S_{n - i}.
$$
Since $b_{n-i+1} \geqslant b_{n-i}$ we can subtract the respective expressions from both sides:
\begin{align*}
&(k - i)\cdot S_{n - i + 1}  - (k-i) \cdot b_{n-i+1} \leqslant \\
&\hspace{5em}(k - i + 1)\cdot  S_{n - i} - (k-i) \cdot b_{n-i}.
\end{align*}
Simplifying the inequality yields:
\begin{align*}
(k-i) \cdot S_{n-i} &\leqslant (k-i) \cdot S_{n-i-1} + S_{n-i} \;\Rightarrow \\
a_i = \frac{S_{n-i}}{k-i} &\leqslant \frac{S_{n-i-1}}{k-i-1} = a_{i+1}
\end{align*}
Thus, we have shown that if $a_{i-1} \leqslant a_i$ holds for some $i$,
then the remaining subsequence is non-decreasing.
Due to the arbitrary choice of $i$ the whole sequence $\left\{a_i\right\}$ is either
strictly decreasing, or it is decreasing up to some index and non-decreasing thereafter.
Therefore, $\left\{a_i\right\}$ is unimodal with unique global minimum.
\end{proof}

\subsection{Constructing of optimal matching for \texorpdfstring{$\PROBLEM(K_n, K_k, b)$}{VMCAP(Kn, Kk, b)}}
\label{appendix:knkksol}
From the proof of Lemma \ref{lemma:knkkrec} we obtain a straightforward greedy
algorithm to construct a solution that achieves $\PROBLEM$ number of matches.
It is an iterative way of matching $k$ vertices with the largest remaining capacities
(thus, progressively reducing imbalance between vertex capacities).
Performance can be improved by introducing batching.
Note that each iteration of the greedy algorithm produces the same match up to the point
when vertices change their position in the sorted list.
This happens when the capacity of the initially $b^{(k)}$ vertex
becomes one unit less than that of the $b^{(k+1)}$ vertex.
This observation allows the creation of matches in batches
(see Algorithm \ref{alg:knkksol}),
which empirically reduces the number of iterations required.

\begin{algorithm}[H]
\caption{Matching construction for $\PROBLEM(K_n, K_k, b)$}
\label{alg:knkksol}
\begin{algorithmic}[1]
\STATE \textbf{Input:} number of pNUMA nodes $n$;\newline
       \phantom{\textbf{Input:}} number of vNUMA nodes $k$;\newline
       \phantom{\textbf{Input:}} list of pNUMA capacities $B = [b_1,\ldots,b_n]$
\STATE \textbf{Output:} list of matches $M=[\{i_1,\ldots,i_k\},\,...]$
\STATE Let $M \gets [\,]$
\WHILE{\TRUE}
    \STATE Let $I \gets [i \;|\; b_i > 0,\; \forall i \in 1..n]$
    \IF{$|I| < k$}
       \RETURN $M$
    \ENDIF
    \STATE Sort $I$ by $B$ values in descending order
    \IF{$|I| > k$}
        \STATE Let $s \gets B[I[k]] - B[I[k\!+\!1]] + 1$
    \ELSE
        \STATE Let $s \gets B[I[k]]$
    \ENDIF
    \FOR{$j=1$ \TO $s$}
        \STATE Append $\{I[1],..,I[k]\}$ to $M$
    \ENDFOR
    \FOR{$j=1$ \TO $k$}
        \STATE $B[I[j]] \gets B[I[j]] - s$
    \ENDFOR
\ENDWHILE
\end{algorithmic}
\end{algorithm}

\bibliographystyle{ieeetr}
\bibliography{bibliography}

\end{document}